\documentclass{llncs}
\usepackage{amsmath}

\begin{document}

   \setlength\abovedisplayskip{0pt}
   \setlength\belowdisplayskip{0pt}

\title{Cost Sharing Security Information with Minimal Release Delay}
\author{}
\institute{}
\author{Mingyu Guo \and
Yong Yang \and
Muhammad Ali Babar
}
\authorrunning{M. Guo et al.}
\institute{The University of Adelaide, Adelaide, Australia
\email{\{mingyu.guo,ali.babar\}@adelaide.edu.au}\\
\email{yong.yang@student.adelaide.edu.au}}
\maketitle

\begin{abstract} We study a cost sharing problem derived from bug bounty
    programs, where agents gain utility by the amount of time they get to enjoy
    the cost shared information. Once the information is provided to an agent, it
    cannot be retracted.  The goal, instead of maximizing revenue, is to pick a
    time as early as possible, so that enough agents are willing to cost share
    the information and enjoy it for a premium time period, while other agents wait
    and enjoy the information for free after a certain amount of release delay. We
    design a series of mechanisms with the goal of minimizing the maximum delay
    and the total delay. Under prior-free settings, our final mechanism
    achieves a competitive ratio of $4$ in terms of maximum delay, against an
    undominated mechanism. Finally, we assume some distributions of the agents'
    valuations, and investigate our mechanism's performance in terms of
expected delays.  \keywords{Mechanism design  \and Cost sharing \and Bug
bounty.} \end{abstract}

\section{Introduction}
The market for software vulnerabilities---also known as bugs---is a crowded
one. For those holding a serious bug to sell, there are many kinds of
interested customers: the software vendors themselves that can produce official
patches, the anonymous buyers in the black markets that boast greater
reward~\cite{Algarni2014:Software}, and many others in between---such as the
vulnerability brokers.

As defined by B\"ohme~\cite{Boehme2006:Comparison}, by vulnerability brokers,
we refer to organizations other than software vendors that purchase
vulnerabilities and produce corresponding defense services (such as intrusion
detection systems~\cite{Howard2009:Cyber}) for their subscribers. Bug bounty
programs offered by vulnerability brokers provide greater financial incentives
for vulnerability sellers, as their customers could include large corporations
and government agencies that have huge budgets for
security~\cite{Howard2009:Cyber}. One common problem these programs have is
that their subscribers are usually charged an annual subscription
fee~\cite{Boehme2006:Comparison}, while they certainly don't produce a constant
number of security updates every year, and each customer may not benefit
equally with each update---for example, an update that helps prevent a bug in
Windows operating system would be of little interest to customers that don't
use Windows at all, though they still have to pay for the update with the fixed
subscription fee.

While this inequality can be trivially solved by designing as many subscription
levels as necessary, we are introducing a game-theoretical model for non-profit
bug bounty programs that both solves this efficiency problem and promotes
general software security.

Specifically, we study the mechanism design problem of selling one bug
(information regarding it) to multiple agents. The goal is not to make a
profit, but we need the mechanism to cover the cost of the bug. All agents
receive the bug if enough payments can be collected to cover the cost. To
incentivize payments, agents who do not pay receive the bug slightly delayed.
Our goal is to maximize the social welfare by minimizing the maximum and the
total delay. We end up with a mechanism that is 4-competitive against an
undominated mechanism in terms of maximum delay, and for expected delay, we conclude by discussing the expected delay under some assumptions on the distributions of the agents' valuations.

Although this problem we are studying is derived from bug bounty programs, it
certainly could relate to other systems. So here we define the traits that
characterize the problem. The service or good that is sold has unlimited supply
once funded (zero marginal cost), and cannot be retracted once given to a user.
The most common examples are information and digital goods. The agents have a
valuation function that is non-decreasing in terms of time: the earlier the
agent gets the information, the more utility she receives. And as we are
designing non-profit systems, the mechanism should be budget balanced: we
charge the agents exactly the amount needed to purchase the bug which the
defense information is derived from. Finally, we want to incentivize enough
payments with long premium time periods (periods exclusively enjoyed by the
paying agents). But we also want the premium time periods to be as short as
possible so that non-paying agents can receive the information sooner, as it
leads to higher social welfare.

\section{Related Research}

With more and more critical software vulnerabilities catching the public's
attention, there's an increasing amount of literature on the market for
vulnerabilities.  However, we failed to identify any research that shares the
same problem structure or the same goal as ours, so the following work is
mostly on understanding the vulnerability market, and inspirations for future
work, rather than what our study is based on.

Regarding the vendor's possible reluctance to accept and fix reported
vulnerabilities responsibly, Canfield et al.~\cite{Canfield:National} made
quite a few recommendations on ways to incentivize vendors to fix the
software's vulnerabilities responsibly, and general improvement suggestions
including allowing negotiations of the severity level of discovered
vulnerabilities; on the subject of how and when should bugs be disclosed to the
general public.  Arora et al.~\cite{Arora2008:Optimal} produced numerical
simulations which suggested instant disclosure of vulnerabilities to be
sub-optimal. Nizovtsev and Thursby~\cite{Nizovtsev2007:disclose}, unlike
others, used a game-theoretic approach to show that full disclosure can be an
equilibrium strategy, and commented on the pressure of instant disclosure may
put on vendors may have a long-term effect that improves software quality.
Also, there had been discussions on the feasibility of introducing markets for
trading bugs openly~\cite{Miller2007:legitimate}, with some going as far as
designing revenue-maximizing mechanisms for
them~\cite{Guo2017:Optimizing,Guo2016:Revenue}.

Then, when introducing new bug bounty programs, it's quite necessary to
consider its effect outside the expected producer and consumer population.
Maillart et al.~\cite{Maillart2017:Given} proved that each newly launched program has
a negative effect on submissions to existing bug bounty programs, and they also
analyzed the researchers (bounty hunters) expected gains and strategies in
participating in bug bounty programs. Specifically for vulnerability
brokers, Kannan et al.~\cite{Kannan2005:Market} emphasized a caveat that a
vulnerability broker (which is called a market-based infomediary in their
paper) always has incentive to leak the actual vulnerability, as ``\ldots This
leakage exposes nonsubscribers to attacks from the hacker. The leakage also
serves to increase the users’ incentives to subscribe to the infomediary's
service.''

Finally, we found a sorely lacking amount of literature on existing
vulnerability brokers and the actual sellers of vulnerabilities. Although a
few papers on these topics were located~\cite{Howard2009:Cyber,Guo2017:Optimizing,Guo2016:Revenue}, we
did not find any detailed models or discussions, perhaps due to the secretive
nature of the cybersecurity business.

\section{Model Description}\label{sec:model}

We study the problem of selling one bug (with a fixed cost) to $n$ agents.  Our
goal is \emph{not} to make a profit, but we need the mechanism to cover the
cost of the bug.  Without loss of generality, we assume the cost of the bug is
$1$.

Our mechanism would generally charge a total payment of $1$ from the agents (or
charge $0$, in which case the bug is not sold). If the bug is sold, then we
provide the bug to \emph{all} agents, including those who pay very little
or do not pay at all. There are a few reasons for this design decision:

\begin{itemize}

    \item The main goal of this non-profit system is to promote general
        software security, so we would like to have as many people protected
        from the vulnerability as possible.

    \item Since no cost is incurred in distributing the bug once funded,
        the system and the agents don't lose anything by allowing the presence
        of free riders.

    \item In practise, providing free security information encourages more
        agents to join the system. Under our cost sharing mechanism, including
        more agents actually generally leads to less individual payment and
        increased utilities for everyone.

\end{itemize}

To incentivize payments, if an agent has a higher valuation (is willing to pay
more), then our mechanism provides the bug information to this agent slightly
earlier.  For the free riders, they receive the bug for free, except that there
is a bit of delay. Our aim is to minimize the delay (we cannot completely get
rid of the delay as it is needed for collecting payments).

We assume the bug has a life cycle of $[0, 1]$. Time $0$ is when the sale
starts. Time $1$ is when the bug reaches its end of life cycle (or when the bug
becomes public knowledge).

We use $v_i$ to denote agent $i$'s type. If agent $i$ receives the bug at time
$t$, then her valuation equals $(1-t)v_i$. That is, if she receives the bug at
time $0$, then her valuation is simply $v_i$.  If she receives the bug at time
$1$, then her valuation is $0$.

We use $t_i^M(v_i,v_{-i})$ and $p_i^M(v_i,v_{-i})$ to denote agent $i$'s
allocation time and payment, under mechanism $M$, when agent $i$ reports $v_i$
and the other agents report $v_{-i}$.\footnote{For randomized mechanisms, the
allocation times and payments are the expected values over the random bits.}
Agent $i$'s utility $u_i^M(v_i,v_{-i})$ is
$(1-t_i^M(v_i,v_{-i}))v_{i}-p_i^M(v_i,v_{-i})$.

We enforce three mechanism constraints in this paper:
\emph{strategy-proofness}, \emph{individual rationality}, and \emph{ex post
budget balance}.  They are formally defined as follows:

\begin{itemize}

    \item Strategy-proofness: for any $v_i, v_i', v_{-i}$, \[
            (1-t_i^M(v_i,v_{-i}))v_i-p_i^M(v_i,v_{-i})\ge
            (1-t_i^M(v_i',v_{-i}))v_i-p_i^M(v_i',v_{-i}) \]

    \item Individual Rationality: for any $v_i, v_{-i}$,
        \[ (1-t_i^M(v_i,v_{-i}))v_i-p_i^M(v_i,v_{-i})\ge 0\]

    \item Ex post budget balance\footnote{For randomized mechanisms, we require
        that for all realizations of the random bits, the constraint holds.}:

        \begin{itemize}

            \item[] If the bug is sold, then we must have
                \[\sum_{i}p_i^M(v_i,v_{-i})=1\]

            \item[] If the bug is not sold, then we must have that for all $i$
                \[p_i^M(v_i,v_{-i})=0\ \textnormal{and}\ t_i^M(v_i,v_{-i})=1\]

        \end{itemize}

\end{itemize}

We study the minimization of two different mechanism design objectives.
The \textsc{Max-Delay} and \textsc{Sum-Delay} are defined as follows:
\begin{align*}
    \textsc{Max-Delay:}\ \max_{i}t_i^M(v_i,v_{-i})\\
    \textsc{Sum-Delay:}\ \sum_{i}t_i^M(v_i,v_{-i})
\end{align*}

Our setting is a single-parameter setting where Myerson's characterization applies.

\begin{claim}[Myerson's Characterization~\cite{Myerson1981:Optimal}] Let $M$ be a strategy-proof and
    individually rational mechanism, we must have that

    \begin{itemize}

        \item For any $i$ and $v_{-i}$, $t_i^M(v_i,v_{-i})$ is non-increasing
            in $v_i$.  That is, by reporting higher, an agent's allocation time
            never becomes later.

        \item The agents' payments are completely characterized by the
            allocation times. That is, $p_i^M$ is determined by $t_i^M$.
            \[p_i^M(v_i, v_{-i}) =
            v_i(1-t_i^M(v_i,v_{-i}))-\int_{z=0}^{v_i}(1-t_i^M(z,
            v_{-i}))\,\mathrm{d}z \] The above payment characterization also
            implies that both the payment $p_i^M(v_i,v_{-i})$ and the utility
            $u_i^M(v_i,v_{-i})$ are non-decreasing in $v_i$.

    \end{itemize}

\end{claim}

\section{Prior-Free Settings}

In this section, we focus on problem settings where we do not have the prior
distributions over the agents' types.  For both \textsc{Max-Delay} and
\textsc{Sum-Delay}, the notion of optimal mechanism is not well-defined.  Given
two mechanisms $A$ and $B$, mechanism $A$ may outperform mechanism $B$ under
some type profiles, and vice versa for some other type profiles.

We adopt the following dominance relationships for comparing mechanisms.

\begin{definition}

    Mechanism $A$ \textsc{Max-Delay-Dominates} mechanism $B$, if
    \begin{itemize}

        \item for \emph{every} type profile, the \textsc{Max-Delay} under
            mechanism $A$ is \emph{at
            most}\footnote{\label{footnote:tie-breaking}Tie-breaking detail:
            given a type profile, if under $A$, the bug is not sold (max delay
            is $1$), and under $B$, the bug is sold (the max delay happens to
            be also $1$), then we interpret that the max delay under $A$ is
            \emph{not} at most that under $B$.} that under mechanism $B$.

        \item for \emph{some} type profiles, the \textsc{Max-Delay} under
            mechanism $A$ is \emph{less than} that under mechanism $B$.

    \end{itemize}

    A mechanism is \textsc{Max-Delay-Undominated}, if it is not dominated by
    any \emph{strategy-proof} and \emph{individually rational} mechanisms.

\end{definition}

\begin{definition}

    Mechanism $A$ \textsc{Sum-Delay-Dominates} mechanism $B$, if

    \begin{itemize}

        \item for \emph{every} type profile, the \textsc{Sum-Delay} under
            mechanism $A$ is \emph{at most} that under mechanism $B$.

        \item for \emph{some} type profiles, the \textsc{Sum-Delay} under mechanism
            $A$ is \emph{less than} that under mechanism $B$.

    \end{itemize}

    A mechanism is \textsc{Sum-Delay-Undominated}, if it is not dominated by
    any \emph{strategy-proof} and \emph{individually rational} mechanisms.

\end{definition}

For our model, one trivial mechanism works as follows:

\vspace{.1in}
\noindent\framebox{\parbox{\textwidth}{%
    \begin{center}
    \textbf{Cost Sharing (CS)}
    \end{center}

    \flushright{%
    Strategy-proofness: Yes

    Individual rationality: Yes

    Ex post budget balance: Yes}

    \begin{itemize}

        \item Consider the following set: \[ K = \{ k\ \vert\ \textnormal{$k$
            values among the $v_i$ are at least $1/k$}, 1\le k \le n \} \]

        \item If $K$ is empty, then the bug is not sold. Every agent's
            allocation time is $1$ and pays $0$.

        \item If $K$ is not empty, then the highest $k^* = \max K$ agents
            each pays $1/k^*$ and receives the bug at time $0$. The other
            agents receive the bug at time $1$ and each pays $0$.

    \end{itemize}
}}
\vspace{.1in}

The above mechanism is strategy-proof, individually rational, and ex post
budget balanced.  Under the mechanism, $k^*$ agents join in the cost sharing
and their delays are $0$s, but the remaining agents all have the maximum delay
$1$. Both the \textsc{Max-Delay} and the \textsc{Sum-Delay} are bad when $k^*$
is small.  One natural improvement is as follows:

\vspace{.1in}
\noindent\framebox{\parbox{\textwidth}{%
    \begin{center}
    \textbf{Cost Sharing with Deadline (CSD)}
    \end{center}

    \flushright{%
    Strategy-proofness: Yes

    Individual rationality: Yes

    Ex post budget balance: No}

    \begin{itemize}

        \item Set a constant deadline of $0\le t_C \le 1$. Under the mechanism,
            an agent's allocation time is at most $t_C$.

        \item Consider the following set: \[ K = \{ k\ \vert\ \textnormal{$k$
            values among the $v_i$ are at least $\frac{1}{kt_C}$}, 1\le k \le n
            \} \]

        \item If $K$ is empty, then the bug is not sold. Every agent's
            allocation time is $t_C$ and pays $0$.

        \item If $K$ is not empty, then the highest $k^* = \max K$ agents
            each pays $1/k^*$ and receives the bug at time $0$. The other
            agents receive the bug at time $t_C$ and each pays $0$.

    \end{itemize}
}}
\vspace{.1in}

The idea essentially is that we run the trivial cost sharing (\textsc{CS}) mechanism
on the time interval $[0,t_C]$, and every agent receives the time interval $[t_C,1]$
\emph{for free}.   The mechanism remains strategy-proof and
individually rational.  Unfortunately, the mechanism is not ex post budget
balanced---even if the cost sharing failed (\emph{e.g.}, $K$ is empty), we
still need to reveal the bug to the agents at time $t_C$ for free. If $t_C<1$,
we have to pay the seller without collecting back any payments.

The reason we describe the \textsc{CSD} mechanism is because our final
mechanism uses it as a sub-component, and the way it is used fixes the budget
balance issue.








\begin{example}
    Let us consider the type profile $(0.9, 0.8, 0.26, 0.26)$. We run the cost sharing
    with deadline (\textsc{CSD}) mechanism using different $t_C$ values:

    \begin{itemize}

        \item If we set $t_C=0.9$, then agent $1$ and $2$ would receive the bug
            at time $0$ and each pays $0.5$. Agent $3$ and $4$ pay nothing but they have to wait until time $0.9$.

        \item If we set $t_C=0.7$, then agent $1$ and $2$ would still receive the bug
            at time $0$ and each pays $0.5$. Agent $3$ and $4$ pay nothing but they only
            need to wait until $0.7$. This is obviously better.

        \item If we set $t_C=0.5$, then all agents pay $0$ and only wait until
            $0.5$.  However, we run into budget issue in this scenario.

    \end{itemize}

We need $t_C$ to be small, in order to have shorter delays.  However, if $t_C$
    is too small, we have budget issues. The optimal $t_C$ value depends on the
    type profile.  For the above type profile, the optimal
    $t_C=\frac{0.5}{0.8}=0.625$.  When $t_C=0.625$, agent $2$ is still willing
    to pay $0.5$ for the time interval $[0,0.625]$ as $0.8\times 0.625=0.5$.
\end{example}

\begin{definition}

    Given a type profile $(v_1,v_2,\dots,v_n)$, consider the following set:

    \[ K(t_C) = \{ k\ \vert\ \textnormal{$k$ values among the $v_i$ are at least
    $\frac{1}{kt_C}$}, 1\le k \le n \} \]

    $t_C$ is between $0$ and $1$. As $t_C$ becomes smaller, the set $K(t_C)$
    also becomes smaller.  Let $t_C^*$ be the minimum value so that $K(t_C^*)$
    is not empty.  If such $t_C^*$ does not exist (\emph{e.g.}, $K(1)$ is
    empty), then we set $t_C^*=1$.

    $t_C^*$ is called the \textbf{optimal deadline} for this type profile.

\end{definition}

Instead of using a constant deadline, we may pick the optimal deadline for
every type profile.

\vspace{.1in}
\noindent\framebox{\parbox{\textwidth}{%
    \begin{center}
    \textbf{Cost Sharing with Optimal Deadline (CSOD)}
    \end{center}

    \flushright{%
    Strategy-proofness: No

    Individual rationality: Yes

    Ex post budget balance: Yes}

    \begin{itemize}
        \item For every type profile, we calculate its optimal deadline.
        \item We run \textsc{CSD} using the optimal deadline.
    \end{itemize}
}}
\vspace{.1in}

\textsc{CSOD} is ex post budget balanced. If we cannot find $k$ agents to pay
$1/k$ each for any $k$, then the optimal deadline is $1$ and the cost sharing
failed. That is, we simply do not reveal the bug (choose not to buy the bug
from the seller).

Unfortunately, we gained some and lost some. Due to changing deadlines, the
mechanism is not strategy-proof.

\begin{example} Let us re-consider the type profile $(0.9, 0.8, 0.26, 0.26)$.
    The optimal deadline for this type profile is $0.625$.  By reporting
    truthfully, agent $2$ receives the bug at time $0$ and pays $0.5$.
    However, she can lower her type to $0.26$ (the optimal deadline is now
slightly below $1$). Agent $2$ still receives the bug at time $0$ but only pays
$0.25$.  \end{example}

Other than not being strategy-proof, under our prior-free settings,
\textsc{CSOD} is optimal in the following senses:

\begin{theorem} Cost sharing with optimal deadline (\textsc{CSOD}) is both
\textsc{Max-Delay-Undominated} and \textsc{Sum-Delay-Undominated}.
\end{theorem}

\begin{proof}

    We first focus on \textsc{Max-Delay-Undominance}.  Let $M$ be a
    strategy-proof and individually rational mechanism that
    \textsc{Max-Delay-Dominates} \textsc{CSOD}.
    We will prove by contradiction
    that such a mechanism does not exist.

    Let $(v_1,v_2,\dots,v_n)$ be an arbitrary type profile. Without loss of
    generality, we assume $v_1\ge v_2\ge\dots v_n$. We will show that $M$'s
    allocations and payments must be identical to that of \textsc{CSOD} for
    this type profile. That is, $M$ must be identical to \textsc{CSOD}, which
    results in a contradiction.

    We first consider type profiles under which the bug is sold under
    \textsc{CSOD}.  We still denote the type profile under discussion by
    $(v_1,v_2,\dots,v_n)$. Let $k^*$ be the number of agents who participate in
    the cost sharing under \textsc{CSOD}.

    We construct the following type profile:

    \begin{equation}\label{tp:trivial}
        (\underbrace{1/k^*, \dots,1/k^*}_{k^*},0,\dots,0)
    \end{equation}

    For the above type profile, under \textsc{CSOD}, the first $k^*$ agents
    receive the bug at time $0$ and each pays $1/k^*$. By dominance assumption
    (both \textsc{Max-Delay-Dominance} and \textsc{Sum-Delay-Dominance}), under
    $M$, the bug must also be sold. To collect $1$, the first $k^*$ agents must
    each pays $1/k^*$ and must receive the bug at time $0$ due to individual
    rationality.

    Let us then construct a slightly modified type profile:

    \begin{equation}\label{tp:v1}
    (v_1,\underbrace{1/k^*, \dots,1/k^*}_{k^*-1},0,\dots,0)
    \end{equation}

    Since $v_1\ge 1/k^*$, under $M$, agent $1$ must still receive the bug at
    time $0$ due to the monotonicity constraint.  Agent $1$'s payment must
    still be $1/k^*$. If the new payment is lower, then had agent $1$'s true
    type been $1/k^*$, it is beneficial to report $v_1$ instead. If the new
    payment is higher, then agent $1$ benefits by reporting $1/k^*$ instead.
    Agent $2$ to $k^*$ still pay $1/k^*$ and receive the bug at time $0$ due to
    individual rationality.

    We repeat the above reasoning by constructing another slightly modified
    type profile:

    \begin{equation}\label{tp:v2}
    (v_1,v_2,\underbrace{1/k^*, \dots,1/k^*}_{k^*-2},0,\dots,0)
    \end{equation}

    Due to the monotonicity constraint, agent $2$ still pays $1/k^*$ and
    receives the bug at time $0$. Had agent $1$ reported $1/k^*$, he would
    receive the bug at time $0$ and pay $1/k^*$, so due to the monotonicity
    constraint, agent $1$ still pays $1/k^*$ and receives the bug at time $0$
    under type profile~\eqref{tp:v2}. The rest of the agents must be
    responsible for the remaining $(k^*-2)/k^*$, so they still each pays
    $1/k^*$ and receives the bug at time $0$.

    At the end, we can show that under $M$, for the following type profile, the
    first $k^*$ agents each pays $1/k^*$ and must receive the bug at $0$.

    \begin{equation}\label{tp:vk}
        (v_1,v_2,\dots,v_{k^*},0,\dots,0)
    \end{equation}

    For the above type profile~\eqref{tp:vk}, there are $n-k^*$ agents
    reporting $0$s.  For such agents, their payments must be $0$ due to
    individual rationality. Since $M$ \textsc{Max-Delay-Dominates}\footnote{The
    claim remains true if we switch to \textsc{Sum-Delay-Dominance}.}
    \textsc{CSOD}, these agents' allocation time must be at most
    $\frac{1}{k^*v_{k^*}}$, which is their allocation time under \textsc{CSOD}
    (this value is the optimal deadline).  We show that they cannot receive the
    bug strictly earlier than $\frac{1}{k^*v_{k^*}}$ under $M$.

    Let us consider the following type profile:

    \begin{equation}\label{tp:vk1}
        (v_1,v_2,\dots,v_{k^*},\frac{k^*v_{k^*}}{k^*+1},\dots,0)
    \end{equation}

For type profile~\eqref{tp:vk1}, agent $k^*+1$ must receive the bug at time
    $0$ and pay $1/(k^*+1)$.  She can actually benefit by reporting $0$
    instead, if under type profile~\eqref{tp:vk}, agents reporting $0$ receive
    the bug at $\frac{1}{k^*v_{k^*}}^* <\frac{1}{k^*v_{k^*}}$ for free.
   \begin{align*}
       \textnormal{utility for reporting truthfully} &= \frac{k^*v_{k^*}}{k^*+1}-\frac{1}{k^*+1},
       \textnormal{utility for reporting $0$} =\\ \frac{k^*v_{k^*}}{k^*+1}(1-\frac{1}{k^*v_{k^*}}^*) &> \frac{k^*v_{k^*}}{k^*+1}(1-\frac{1}{k^*v_{k^*}}) = \frac{k^*v_{k^*}}{k^*+1}-\frac{1}{k^*+1}
    \end{align*}

    Therefore, for type profile~\eqref{tp:vk}, all agents who report $0$ must
    receive the bug at exactly $\frac{1}{k^*v_{k^*}}$. That is, for type
    profile~\eqref{tp:vk}, $M$ and \textsc{CSOD} are equivalent.

    Now let us construct yet another modified type profile:

    \begin{equation}\label{tp:vk1true}
        (v_1,v_2,\dots,v_{k^*},v_{k^*+1},0,\dots,0)
    \end{equation}

    Here, we must have $v_{k^*+1}<\frac{k^*v_{k^*}}{k^*+1}$. Otherwise, under
    the original type profile, we would have more than $k^*$ agents who join
    the cost sharing. We assume under $M$, agent $k^*+1$ receives the bug at
    time $t$ and pays $p$. $t$ is at most $\frac{1}{k^*v_{k^*}}$ due to the
    monotonicity constraint. We have

    \begin{align*}
     \textnormal{utility when the true type is $v_{k^*+1}$ and reporting truthfully} &=v_{k^*+1}(1-t)-p\\
     \textnormal{utility when the true type is $v_{k^*+1}$ and reporting $0$} &= v_{k^*+1}(1-\frac{1}{k^*v_{k^*}})
    \end{align*}

    Therefore,
    \begin{equation}\label{eq:oneside}
    v_{k^*+1}(1-t)-p\ge v_{k^*+1}(1-\frac{1}{k^*v_{k^*}})
    \end{equation}

    Had agent $k^*+1$'s type been $\frac{k^*v_{k^*}}{k^*+1}$, her utility for reporting
    her true type must be at least her utility when reporting $v_{k^*+1}$. That is,

    \begin{align*}
    \textnormal{utility when the true type is $\frac{k^*v_{k^*}}{k^*+1}$ and reporting truthfully} &= \frac{k^*v_{k^*}}{k^*+1}-\frac{1}{k^*+1} \\
        \textnormal{utility when the true type is $\frac{k^*v_{k^*}}{k^*+1}$ and reporting $v_{k^*+1}$} &= \frac{k^*v_{k^*}}{k^*+1}(1-t) -p
    \end{align*}

    That is,
    \begin{equation}\label{eq:twoside}
    \frac{k^*v_{k^*}}{k^*+1}-\frac{1}{k^*+1} \ge \frac{k^*v_{k^*}}{k^*+1}(1-t) -p
    \end{equation}

    Combine Equation~\eqref{eq:oneside}, Equation~\eqref{eq:twoside},
    $v_{k^*+1}<\frac{k^*v_{k^*}}{k^*+1}$, and $t\le\frac{1}{k^*v_{k^*}}$, we
    have $p=0$ and $t=\frac{1}{k^*v_{k^*}}$.  That is, under type
    profile~\eqref{tp:vk1true}, agent $k^*+1$'s allocation and payment remain
    the same whether she reports $0$ or $v_{k^*+1}$.

    Repeat the above steps, we can show that under the following arbitrary
    profile, agent $k^*+2$ to $n$'s allocation and payment also remain the same
    as when they report $0$.

    \begin{equation}\label{tp:vkn}
        (v_1,v_2,\dots,v_{k^*},v_{k^*+1},v_{k^*+2},\dots,v_n)
    \end{equation}

    That is, for type profiles where the bug is sold under \textsc{CSOD}, $M$
    and \textsc{CSOD} are equivalent.

    We then consider an arbitrary type profile for which the bug is not sold
    under \textsc{CSOD}.  Due to the monotonicity constraint, an agent's
    utility never decreases when her type increases. If any agent $i$ receives
    the bug at time $t$ that is strictly before $1$ and pays $p$, then due to
    the individual rationality constraint, we have that $v_i(1-t)-p\ge 0$.
    $v_i$ must be strictly below $1$, otherwise the bug is sold under
    \textsc{CSOD}.  Had agent $i$'s true type been higher but still below $1$
    (say, to $v_i+\epsilon$), her utility must be positive, because she can
    always report $v_i$ even when her true type is $v_i+\epsilon$.  But earlier
    we proved that had $v_i$'s true type been $1$, she would receive the bug at
    time $0$ and pay $1$. Her utility is $0$ when her type is $1$.  This means
    her utility decreased if we change her true type from $v_i+\epsilon$ to
    $1$, which is a contradiction. That is, all agents must receive the bug at
    time $1$ (and must pay $0$). Therefore, for an arbitrary type profile for
    which the bug is not sold under \textsc{CSOD}, $M$ still behaves the same.

    In the above proof, all places where we reference
    \textsc{Max-Delay-Dominance} can be changed to
    \textsc{Sum-Delay-Dominance}.
\qed%
\end{proof}

\textsc{CSOD} is both \textsc{Max-Delay-Undominated} and
\textsc{Sum-Delay-Undominated}, but it is not strategy-proof.  We now propose
our final mechanism in this section, which builds on \textsc{CSOD}.  The new
mechanism is strategy-proof and its delay is within a constant factor of
\textsc{CSOD}.\footnote{That is, we fixed the strategy-proofness issue at the
cost of longer delays, but it is within a constant factor.}

\vspace{.1in}
\noindent\framebox{\parbox{\textwidth}{%
    \begin{center}
    \textbf{Group-Based Cost Sharing with Optimal Deadline (GCSOD)}
    \end{center}

    \flushright{%
    Strategy-proofness: Yes

    Individual rationality: Yes

    Ex post budget balance: Yes}

    \begin{itemize}
        \item For agent $i$, we flip a fair coin to randomly assign her to either the left group or the right group.
        \item We calculate the optimal deadlines of both groups.
        \item We run \textsc{CSD} on both groups.
        \item The left group uses the optimal deadline from the right group and vice versa.
    \end{itemize}
}}
\vspace{.1in}

\begin{claim} Group-based cost sharing with optimal deadline (\textsc{GCSOD})
is strategy-proof, individually rational, and ex post budget balanced.
\end{claim}

\begin{proof}
    Every agent participates in a \textsc{CSD} so strategy-proofness and
    individual rationality hold. Let $D_L$ and $D_R$ be the optimal deadlines
    of the left and right groups, respectively. If $D_L < D_R$, then the left
    group will definitely succeed in the cost sharing, because its optimal
    deadline is $D_L$ and now they face an extended deadline. The right group
    will definitely fail in the cost sharing, as they face a deadline that is
    earlier than the optimal one.  At the end, some agents in the left group
    pay and receive the bug at $0$, and the remaining agents in the left group
    receive the bug at time $D_R$ for free.  All agents in the right group
    receive the bug at time $D_L$ for free.  If $D_L > D_R$, the reasoning is
    the same. If $D_L=D_R < 1$, then we simply tie-break in favour of the left
    group. If $D_L=D_R=1$, then potentially both groups fail in the cost
    sharing, in which case, we simply do not reveal the bug (do not buy it from
    the seller).
\qed%
\end{proof}

\begin{definition} Mechanism $A$ is \textsc{$\alpha$-Max-Delay-Competitive}
    against mechanism $B$ if for every agent $i$, every type profile, we have
    that the max delay under $A$ is at most $\alpha$ times the max delay under
    $B$.

    \textsc{$\alpha$-Sum-Delay-Competitive} is defined similarly.
\end{definition}

\begin{theorem}\label{th:4}
    \textsc{GCSOD} is \textsc{$4$-Max-Delay-Competitive} against \textsc{CSOD} under
    two technical assumptions:
    \begin{itemize}
        \item No agent's valuation for the bug exceeds the whole cost. That is, $v_i\le 1$ for all $i$.
        \item At least one agent does not participate in the cost sharing under \textsc{CSOD}.
    \end{itemize}
\end{theorem}

There's no way to ensure that the first assumption always holds, but it can be
argued that it at least holds in the scenarios of cost sharing serious bugs
beyond any individual's purchasing power. The second assumption is needed only
because in the single case of everyone joining the cost sharing under
\textsc{CSOD}, the max delay is 0. While under \textsc{GCSOD}, the max delay is
always greater than 0 so it would not be competitive in this case only. And for
the other assumption, as our system would welcome as many agents as possible,
it is expected that there are always agents who don't value a new bug very much
so that they would prefer to be free riders instead of participating in the
cost sharing under \textsc{CSOD}.

\begin{proof} Let us consider an arbitrary type profile that satisfies both
    technical assumptions. We denote it by $(v_1,v_2,\dots,v_n)$.  Without loss
    of generality, we assume $v_1\ge v_2\ge\dots\ge v_n$.  Let $k^*$ be the
    number of agents who join the cost sharing under \textsc{CSOD}.  The
    optimal deadline under \textsc{CSOD} is then $D^*=\frac{1}{k^*v_{k^*}}$,
    which is exactly the max delay for this type profile.

    Under a specific random grouping, for the set of agents from $1$ to $k^*$,
    we assume $k_L$ agents are assigned to the left group and $k_R=k^*-k_L$
    agents are assigned to the right group.

    For the left group, the optimal deadline is at most
    $\frac{1}{k_{L}v_{k^*}}$ if $k_L\ge 1$, which is at most
    $\frac{k^*}{k_{L}}D^*$.  When $k_L=0$, the optimal deadline is at most $1$.
    Under \textsc{CSOD}, since all types are at most $1$, the optimal deadline
    $D^*$ is at least $1/k^*$.  That is, if $k_L=0$, the optimal deadline of
    the left group is at most $k^*D^*$.

    In summary, the optimal deadline of the left group is at most
    $\frac{k^*}{k_L}D^*$ if $k_L\ge 1$ and $k^*D^*$ if $k_L=0$.  That is, the
    optimal deadline of the left group is at most $\frac{k^*}{\max\{1,
    k_L\}}D^*$

    Similarly, the optimal deadline of the right group is at most
    $\frac{k^*}{\max\{1, k_R\}}D^*$

    The max delay under \textsc{GCSOD} is at most the worse of these two
    deadlines.  The ratio between the max delay under \textsc{GCSOD} and the
    max delay under \textsc{CSOD} is then at most $\frac{k^*}{\max\{1,
    \min\{k_L, k^*-k_L\}\}}$.

    We use $\alpha(k)$ to denote the expected ratio (expectation with regard to
    the random groupings):

    \begin{equation}\label{eq:alphak}
    \alpha(k)=\sum_{k_L=0}^{k}\frac{1}{2^{k}}{k\choose k_L}\frac{k}{\max\{1, \min\{k_L, k-k_L\}\}}
    \end{equation}

    We define $\beta(k)=\alpha(k)2^k$.
    \begin{equation*}
            \beta(k)=\sum_{k_L=0}^{k}{k\choose k_L}\frac{k}{\max\{1, \min\{k_L, k-k_L\}\}}
            =\sum_{k_L=1}^{k-1}{k\choose k_L}\frac{k}{\min\{k_L, k-k_L\}} + 2k
    \end{equation*}
    If $k$ is even and at least $50$, then
        \begin{align*}
            \beta(k)
            &=
            \sum_{k_L=1}^{k/2-1}{k\choose k_L}\frac{k}{\min\{k_L, k-k_L\}}
            +\sum_{k_L=k/2+1}^{k-1}{k\choose k_L}\frac{k}{\min\{k_L, k-k_L\}}
            + 2{k \choose k/2} + 2k \\
            &=
            2\sum_{k_L=1}^{k/2-1}{k\choose k_L}\frac{k}{k_L}
            + 2{k \choose k/2}+ 2k \\
        \end{align*}
        \begin{align*}
            \beta(k)
            &=
            2\sum_{k_L=1}^{k/2-1}{k+1\choose k_L+1}\frac{(k_L+1)k}{(k+1)k_L}
            + 2{k \choose k/2}+ 2k \\
            &\le
            4\sum_{k_L=1}^{k/2-3}{k+1\choose k_L+1} +
            2{k+1\choose k/2-1}\frac{(k/2-1)k}{(k+1)(k/2-2)}\\
            &+
            2{k+1\choose k/2}\frac{(k/2)k}{(k+1)(k/2-1)} +
            2{k+1 \choose k/2}+ 2k \\
            &\le
            4\sum_{k_L=1}^{k/2-3}{k+1\choose k_L+1} +
            2.1{k+1\choose k/2-1} + 4.1{k+1\choose k/2}+ 2k
        \end{align*}
        The ratio between ${k+1\choose k/2}$ and ${k+1\choose k/2-1}$ is at
        most $1.08$ when $k$ is at least $50$.
        \begin{equation*}
\beta(k) \le 4\sum_{k_L=1}^{k/2-3}{k+1\choose k_L+1} + 4{k+1\choose k/2-1} + 4{k+1\choose k/2}+ 2k \le 4\sum_{k_L=0}^{k/2-1}{k+1\choose k_L+1} \le 4\times 2^k
        \end{equation*}




We omit the similar proof when $k$ is odd. In summary, we have $\alpha(k)\le 4$ when $k\ge 50$. For smaller $k$, we
numerically calculated $\alpha(k)$. All values are below $4$.
\qed%
\end{proof}

\begin{corollary}
    \textsc{GCSOD} is \textsc{$8$-Sum-Delay-Competitive} against \textsc{CSOD} under two
    technical assumptions:
    \begin{itemize}
        \item No agent's valuation for the bug exceeds the whole cost. That is, $v_i\le 1$ for all $i$.
        \item At least half of the agents do not participate in the cost sharing under \textsc{CSOD}.
    \end{itemize}
\end{corollary}

\begin{proof}
    Let $D^*$ and $k^*$ be the optimal deadline and the number of agents who
    join the cost sharing under \textsc{CSOD}.  The \textsc{Sum-Delay} of the
    agents under \textsc{CSOD} is $(n-k^*)D^*$.  Under \textsc{GCSOD}, the
    deadlines are at most $4D^*$ according to Theorem~\ref{th:4}.  The
    \textsc{Sum-Delay} is then at most $4D^*n$.  Therefore, the competitive
    ratio is $\frac{4n}{n-k^*}$, which is at least $8$ if $k^*\le n/2$.
\qed%
\end{proof}

\section{Settings with Prior Distributions}\label{sec:withprior}

In this section, we assume that there is a publicly known prior distribution
over the agents' types.  Specifically, we assume that every agent's type is
drawn from an identical and independent distribution, whose support is $[0,U]$.
We still enforce the same set of mechanism constraints as before, namely,
strategy-proofness, individually rationality, and ex post budget balance. Our
aim is to minimize the \emph{expected} \textsc{Max-Delay} or the
\emph{expected} \textsc{Sum-Delay}.  Our main results are two linear programs
for computing the lower bounds of \emph{expected} \textsc{Max-Delay} and
\emph{expected} \textsc{Sum-Delay}.  We then compare the performance of
\textsc{CS} and \textsc{GCSOD} against these lower bounds.

The key idea to obtain the lower bounds is to relax the ex post budget balance
constraint to the following:

\begin{itemize}

    \item With probability $\bbbc$, the bug is not sold under the optimal
        mechanism. $\bbbc$ depends on both the mechanism and the distribution.

    \item Every agent's expected payment is then $(1-\bbbc)/n$, as the agents'
        distributions are symmetric.\footnote{It is without loss of generality
        to assume that the optimal mechanism does not treat the agents
        differently based on their identities. Given a non-anonymous mechanism,
        we can trivially create an ``average'' version of it over all permutations
        of the identities~\cite{Guo2013:Undominated}.  The resulting mechanism
        is anonymous and has the same \textsc{Max-Delay} and
        \textsc{Sum-Delay}.}

    \item Every agent's expected allocation time is at least $\bbbc$, as the
        allocation time is $1$ with probability $\bbbc$.

\end{itemize}

We divide the support of the type distribution $[0,D]$ into $H$ equal segments.
We use $\delta$ to denote $D/H$. The $i$-th segment is then $[(i-1)\delta,
i\delta]$.  Noting that the agents' distributions are symmetric, we do not need
to differentiate the agents when we define the following notation.  We use
$t_i$ to denote an agent's expected allocation time when her type is $i\delta$.
That is, $t_0$ is an agent's expected allocation time when her type is $0$, and
$t_H$ is her expected allocation time when her type is $D$. Similarly, we use
$p_i$ to denote an agent's expected payment when her type is $i\delta$. The
$t_i$ and the $p_i$ are the variables in our linear programming models.

Due to Myerson's characterization, the $t_H$ must be non-increasing. That is,

\[
1\ge t_0 \ge t_1\ge \dots\ge t_H \ge 0
\]

We recall that strategy-proofness and individual rationality together imply
that the agents' payments are completely characterized by the allocation times.
Using notation from Section~\ref{sec:model}, we have

\[p_i^M(v_i, v_{-i}) = v_i(1-t_i^M(v_i,v_{-i}))-\int_{z=0}^{v_i}(1-t_i^M(z,
v_{-i}))\,\mathrm{d}z \]

Using notation from this section, that is

\[
i\delta(1-t_i)-\sum_{z=1}^{i}(1-t_z)\delta \le
p_i \le i\delta(1-t_i)-\sum_{z=0}^{i-1}(1-t_z)\delta
\]

$\bbbc$ is another variable in our linear programming model.  We use $\bbbp(i)$
to denote the probability that an agent's type falls inside the $i$-th interval
$[(i-1)\delta, i\delta]$.  Since every agent's expected payment is
$(1-\bbbc)/n$, we have

\[
    \sum_{z=1}^{H}\bbbp(z)p_{z-1} \le (1-\bbbc)/n \le \sum_{z=1}^{H}\bbbp(z)p_z
\]

Since an agent's expected allocation time is at least $\bbbc$, we have

\[
    \sum_{z=1}^{H}\bbbp(z)t_{z-1} \ge \bbbc
\]

The expected \textsc{Sum-Delay} is at least $\sum_{z=1}^{H}\bbbp(z)t_z$.  We
minimize it to compute a lower bound for the expected \textsc{Sum-Delay}.

To compute a lower bound on the expected \textsc{Max-Delay}, we introduce a few more notations:

\begin{itemize}
    \item Let $A(i)$ be the expected \textsc{Max-Delay} when all agents report higher than $i\delta$.
    \item Let $P^A(i)$ be the probability that all agents report higher than $i\delta$.
    \item Let $B(i)$ be the expected \textsc{Max-Delay} when at least one agent reports at most $i\delta$.
    \item Let $P^B(i)$ be the probability that at least one agent reports at most $i\delta$.
    \item Let $C(i)$ be an agent's expected delay when she reports at most $i\delta$.
\end{itemize}

The expected \textsc{Max-Delay} is at least the following for any $i$:

\[
    A(i)\times P^A(i) + B(i)\times P^B(i)\ge B(i)\times P^B(i)\ge C(i)\times P^B(i)
\]

We minimize~\eqref{obj:max} to compute a lower bound on the expected \textsc{Max-Delay}.

\begin{equation}
    C(i)\times P^B(i) \ge \frac{\sum_{z=1}^{i}t_z\bbbp(z)}{\sum_{z=1}^{i}\bbbp(z)} \times \left\{1-{\left(\sum_{z=i+1}^H\bbbp(z)\right)}^n\right\}\label{obj:max}
\end{equation}

We present the expected delays of \textsc{CS} and \textsc{GCSOD} under different distributions.
$U(0,1)$ refers to the case where every agent's valuation is drawn from the
uniform distribution from $0$ to $1$. $N(0.5, 0.2)$ refers to the case where
every agent's valuation is drawn from the normal distribution with mean $0.5$
and standard devastation $0.2$, conditional on that the value is between $0$
and $1$.

\begin{center}
\begin{tabular}{| c | c | c | c | c | c | c |}
    \hline
    & \multicolumn{3}{c|}{\textsc{Max-Delay}}  & \multicolumn{3}{c|}{\textsc{Sum-Delay}}\\
    \hline
    & \textsc{GCSOD} & \textsc{CS} & Lower Bound & \textsc{GCSOD} & \textsc{CS} & Lower Bound \\
    \hline
    $U(0,1)$, $n=1$ & $1.00$ & $1.00$ & $0.89$ & $1.00$ & $1.00$ & $0.89$ \\
    \hline
    $U(0,1)$, $n=2$ & $0.87$ & $0.75$ & $0.67$ & $1.75$ & $1.50$ & $0.96$ \\
    \hline
    $U(0,1)$, $n=5$ & $0.85$ & $0.67$ & $0.46$ & $2.67$ & $1.41$ & $0.94$ \\
    \hline
    $U(0,1)$, $n=10$ & $0.68$ & $0.65$ & $0.29$ & $3.01$ & $1.13$ & $0.89$ \\
    \hline
    $N(0.5,0.2)$, $n=1$ & $1.00$ & $1.00$ & $0.97$ & $1.00$ & $1.00$ & $0.97$ \\
    \hline
    $N(0.5,0.2)$, $n=2$ & $0.87$ & $0.75$ & $0.63$ & $1.75$ & $1.50$ & $0.89$ \\
    \hline
    $N(0.5,0.2)$, $n=5$ & $0.79$ & $0.27$ & $0.20$ & $2.13$ & $0.40$ & $0.27$ \\
    \hline
    $N(0.5,0.2)$, $n=10$ & $0.54$ & $0.15$ & $0.11$ & $2.20$ & $0.17$ & $0.14$ \\
    \hline
    $N(0.5,0.4)$, $n=1$ & $0.95$ & $0.95$ & $0.92$ & $0.95$ & $0.95$ & $0.92$ \\
    \hline
    $N(0.5,0.4)$, $n=2$ & $0.88$ & $0.76$ & $0.66$ & $1.73$ & $1.48$ & $0.94$ \\
    \hline
    $N(0.5,0.4)$, $n=5$ & $0.84$ & $0.57$ & $0.40$ & $2.54$ & $1.09$ & $0.71$ \\
    \hline
    $N(0.5,0.4)$, $n=10$ & $0.65$ & $0.50$ & $0.26$ & $2.76$ & $0.74$ & $0.59$ \\
    \hline
\end{tabular}
\end{center}

\textsc{CS} outperforms \textsc{GCSOD} in terms of both \textsc{Max-Delay} and
\textsc{Sum-Delay}.  This is not too surprising because \textsc{GCSOD} is
designed for its competitive ratio in the worst case. Our derived lower bounds
show that \textsc{CS} is fairly close to optimality in a lot of cases.

\section{Conclusions and Future Work}
We have come up with a mechanism with competitive ratios of 4 for max delay and 8 for sum delay under certain assumptions. As the problem setting is rather new, there are plenty of options to be explored when designing mechanisms with better performance. Possible solutions showing promise include, for exmaple, another method we considered but did not dedicate as much time into---scheduling fixed prices for different sections of time periods, regardless of the agents' submitted valuations. But such a mechanism will require extensive simulations and analyses to evaluate its performance. It should also be noted that the lack of data for such simulations is to be addressed.

While most of our result is presented under prior-free settings, we made a certain number of assumptions, some of which easily hold true for realistic applications---and therefore rather trivial---some of which less so. For example, there is an assumption that there is at least one agent not participating in the cost sharing in the benchmark function CSOD. This is necessary because we cannot evaluate any mechanism's resulting time against 0 and produce a valid competitive ratio, while this can also be easily satisfied by including free-riders who are determined not to contribute at all. But for the assumption that no agent's valuation exceeds the total required amount, although it is introduced because of similar reasons, we cannot expect it to hold true for every case. So either removing existing constraints to generalize the solution or adding more assumptions to yield better results would be reasonable as immediate future work.
\bibliographystyle{splncs04}

\end{document}